\documentclass[sigconf,authorversion,nonacm]{acmart}

\usepackage{graphicx}
\usepackage{amsmath}
\usepackage{amsthm}
\usepackage{enumitem}
\usepackage{hyperref}
\usepackage[caption=false]{subfig}
\usepackage{url}
\usepackage{microtype}
\usepackage{array}
\usepackage{proof}

\usepackage{relsize}
\usepackage{tikz}

\usepackage{tkz-euclide}
\usetikzlibrary{arrows.meta}

\newtheorem{axiomu}{Axiom}
\newtheorem{axiomc}{Axiom}

\DeclareMathOperator{\Left}{Left}
\DeclareMathOperator{\col}{Col}
\DeclareMathOperator{\p}{Point}
\DeclareMathOperator{\seg}{Segment}

\begin{document}

\title{A Machine-Checked Direct Proof of the Steiner-Lehmus Theorem}

\author{Ariel Kellison}
\orcid{0000-0003-3177-7958}
\affiliation{%
  \institution{Department of Computer Science, Cornell University}
    \city{Ithaca}
  \state{NY}
  \country{USA}}
\email{ak2485@cornell.edu}

\begin{abstract}
A direct proof of the Steiner-Lehmus theorem has eluded geometers for over 170 years. The challenge has been that a proof is only considered direct if it does not rely on \emph{reductio ad absurdum}. Thus, any proof that claims to be direct must show, going back to the axioms, that all of the auxiliary theorems used are also proved directly. In this paper, we give a proof of the Steiner-Lehmus theorem that is guaranteed to be direct. The evidence for this claim is derived from our methodology: we have formalized a constructive axiom set for Euclidean plane geometry in a proof assistant that implements a constructive logic and have built the proof of the Steiner-Lehmus theorem on this constructive foundation.
\end{abstract}



\maketitle

\section{Introduction}
The Steiner-Lehmus theorem, which states that 
\emph{if two internal angle bisectors of a triangle are equal
then the triangle is isosceles,} was posed by C. L. Lehmus in 1840. Since the publication of Jakob Steiner's 1844 proof of the theorem, it has become somewhat infamous for the many failed attempts of a \emph{direct} proof; that is, one that does not use \emph{reductio ad absurdum}. 
\begin{figure}[h!]
\centering
     \includegraphics[scale=1.0]{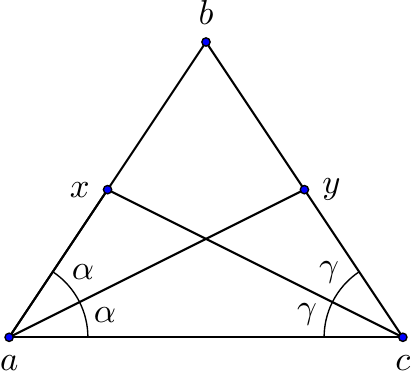}
     \caption{The Steiner-Lehmus theorem: if $ay \cong cx$, $\angle cay \cong \angle yab$, and $\angle bcx \cong xca$ then $ab \cong cb$.}
     \label{fig:SL}
\end{figure}
Numerous allegedly direct proofs have appeared over the years, only to later be discredited due to their reliance on \emph{reductio ad absurdum} by means of auxiliary theorems with indirect proofs. Given the many failures to provide a direct proof, attempts have been made to prove that a direct proof can't possibly exist~\cite{Sylvester,conway_ryba_2014}. Recently, we have been reassured that a direct proof does exist~\cite{Pambuccian}, but we have yet to see one. Thus, the history of the Steiner-Lehmus theorem serves as 177 years of evidence that a human can't account for \emph{all} instances of the use of particular rule of logic, even in the proof of a theorem that many would consider to be rather elementary. 

In this paper, we provide a direct proof of the Steiner-Lehmus theorem. Our guarantee of directness is obtained on modern terms: by formalizing a constructive axiom set for Euclidean plane geometry in the Nuprl proof assistant~\cite{Nuprl} and building a proof of the Steiner-Lehmus theorem on this constructive foundation. 

Our direct proof of the Steiner-Lehmus theorem is given in Section~\ref{sec:SLT} and can be found in the Nuprl library\footnote{The entire formalization of geometry can be found at 
\url{http://www.nuprl.org/LibrarySnapshots/Published/Version2/Mathematics/euclidean!plane!geometry/index.html} and the Steiner-Lehmus theorem can be found at \url{http://www.nuprl.org/LibrarySnapshots/Published/Version2/Mathematics/euclidean!plane!geometry/Steiner-LehmusTheorem.html}.}. Of the many indirect proofs published since the theorem was first posed by Lehmus in 1840, ours is superficially the most similar to the proof given by R.W. Hogg in 1982~\cite{hogg_1982}. Like many of the former proofs of the Steiner-Lehmus theorem, the proof given by Hogg uses \emph{reductio ad absurdum} both explicitly and implicitly; the implicit use is hidden in auxiliary constructions. In comparison, our proof is completely free from the use of \emph{reductio ad absurdum}. 
The atomic relations, axioms, and definitions used in our direct proof of the Steiner-Lehmus theorem are described in Sections~\ref{sec:relations}--\ref{sec:axioms}. The soundness of our axioms with respect to the constructive reals provides additional assurance that the axioms themselves do not harbor any hidden instances of \emph{reductio ad absurdum}. The basis for our model in the constructive reals is given in Section \ref{sec:model}. 

Before introducing the necessary axioms and definitions in Sections~\ref{sec:relations}--\ref{sec:axioms}, we first provide background on the methods of constructive logic that clarify how a direct proof of the Steiner-Lehmus theorem was obtained. 

\section{Constructive Proof, Stability, and Decidability}\label{sec:proof}
When the Steiner-Lehmus theorem was first posed in 1840, the field of logic was not fully formed. Perhaps if it had been, and constructive logic had flourished, geometers would have realized that the Steiner-Lehmus theorem is an example of a \emph{proof of negation}. In particular, the notion of a triangle being isosceles is constructively understood to be a negative statement about a strict notion of inequality of segment lengths. While it is generally assumed that the use of case distinctions is rejected in constructive reasoning, the proof of a negation is an instance where reasoning by cases is constructively valid. 


\emph{Proof by contradiction} (reductio ad absurdum) is a classically admissible reasoning principle that allows one to prove a proposition $P$ by assuming $\neg P$ and deriving absurdity. A \emph{proof of negation} is superficially similar, as one provides proof of the proposition $\neg P$ by assuming $P$ and deriving absurdity. The validity of the two is clearly differentiated in constructive reasoning by the general rejection of the law of double negation elimination for arbitrary propositions $P$: in constructive logic, a proof of negation remains perfectly valid, while a proof by reductio ad absurdum, which is classically equivalent to the law of double negation elimination, does not. 

While the law of double negation elimination is not constructively valid for arbitrary propositions $P$, it is provably true for some propositions. Propositions for which double negation elimination is provably true are referred to as \emph{stable}:
\begin{definition} (Stability) \newline
\begin{center} A proposition $P$ is stable if ${\neg \neg  P \rightarrow \ P }$ holds. \end{center}\label{def:stable}
\end{definition} In constructive logic, all negative propositions are stable. Specifically,
\begin{theorem}\label{thm:negstable}
For all propositions $P$, $\neg \neg \ (\neg P) \rightarrow \neg P$.
\end{theorem}

Furthermore, the proof of a stable proposition permits case distinctions. In particular, the double negation of the law of excluded middle for arbitrary propositions is constructively valid:
\begin{theorem}
For all propositions $P$, $\neg \neg \ ( P \ \lor \neg P)$.
\end{theorem}
Thus, in the proof of a stable proposition $P$, one need only introduce the double negation of the desired case distinction as a hypothesis and apply the appropriate elimination rules to obtain the desired cases. 

Finally, a strictly stronger notion than stability is decidability:
\begin{definition} (Decidability) \newline
\begin{center} A proposition $P$ is decidable if ${P \ \lor  \neg  P}$ holds. \end{center} \label{def:decidable} 
\end{definition}While the decidability of arbitrary propositions does not hold constructively, it is a provable property for many propositions $P$, specifically those for which there is an algorithm for deciding which of $P$ or $\neg P$ holds. 

\subsection{Stable Relations in Constructive Geometry}\label{sec:stable}
We do not take the stability or decidability of any of our atomic or defined relations as axioms. This choice clearly distinguishes our theory from those that take the stability of equality, betweenness, or congruence as axioms~\cite{Beeson2015-BEEACV, MBeeson1,kellison}. Instead, we define equivalence, collinearity, betweenness, and congruence (Definitions \ref{def:eq}, \ref{def:col}, \ref{def:bet}, and \ref{def:cong}) negatively in terms of a strictly positive atomic relation. It follows from Theorem \ref{thm:negstable} that these are stable relations. Our axioms therefore clearly distinguish the geometric propositions that constructively permit case distinctions from those that do not. 

Finally, the choice to not take the stability of equality as an axiom is driven by the desired model, which is the Nuprl implementation of the constructive reals. If we were to take the stability of equality as an axiom, then equality and equivalence would coincide, which does not hold in the Nuprl implementation of the constructive reals~\cite[p.~3]{Bickford1}.

\section{Constructive Geometric Primitives and Relations}\label{sec:relations}

Our axioms rely on two atomic relations on points: a quaternary relation representing an ordering on segment lengths and a ternary relation for plane orientation. These relations are introduced using the formalism of type theory to parallel their implementation in the Nuprl proof assistant. For example, the statement $a:\p$ is to be read as ``$a$ of type $\p$."

\subsection{Segments in Type Theory}\label{sec:seg-and-apart}

The segment type is defined as the Cartesian product of two points. The elements of a Cartesian product are pairs, denoted $\langle a,b \rangle$. If $a$ has type $\p$ and $b$ has type $\p$, then $\langle a,b \rangle$ has type $\p \times \p$: 
\begin{align*}
\infer{\langle a,b \rangle \ : \ \p \times \p}{a : \p & b : \p}.
\end{align*} We will abbreviate segment pairs $\langle a,b \rangle$ by simply writing $ab$:
\begin{align*}
\infer{ab:\seg}{a:\p & b:\p}
\end{align*} When it is necessary to decompose the points constituting a segment $ab$, we may write $\texttt{fst}(ab)$ and $\texttt{snd}(ab)$ for $a$ and $b$ respectively.


\subsection{Atomic Relations and Apartness}

Constructive geometry traditionally utilizes a binary \emph{apartness} relation in place of equality~\cite{HEYTING_axioms-plane-affine-geo,vPlato,vanDalen1996, kellison}. A notable exception is the axiom set presented by Lombard and Vesley~\cite{LOMBARD1998229}, which uses an atomic six place relation and defines a binary apartness relation in terms of the atomic six place relation. In this work, we use an atomic quaternary \emph{strictly greater than} relation to define a binary apartness relation. In particular, given the four points $a,b,c,$ and $d$, if the length of the segment $ab$ is \emph{strictly greater than} the length of the segment $cd$, then the atomic ordering relation on points will be denoted by $ab>cd$. A binary \emph{apartness} relation on points can then be defined using the atomic  \emph{strictly greater than} relation as follows.
\begin{definition}[Apartness of points] The points $a$ and $b$ satisfy an \emph{apartness} relation if the length of the segment $ab$ is strictly greater than the length of the null segment $aa$: \[ a \# b := ab > aa.\]\label{def:papart}
\end{definition} 
The \emph{strictly greater than} relation is used to define two additional quaternary relations on points: apartness of segment lengths and a non-strict ordering of segment lengths. 

 \begin{definition}[Apartness of segment lengths] The length of the segments $ab$ and $cd$ satisfy a \emph{length apartness} relation if either the length of the segment $ab$ is strictly greater than the length of the segment $cd$ or the length of the segment $cd$ is strictly greater than the length of the segment $ab$: \[ ab \# cd := ab > cd \ \lor \ cd > ab.\]\label{def:lenapart} 
\end{definition}

\begin{definition}[Non-strict order of segment lengths] The length of $ab$ is \emph{greater than or equal to} the length of $cd$ if the length of $cd$ is not strictly greater than the length of $ab$:
\begin{align*}ab\ge cd : = \neg cd > ab. \end{align*}\label{def: non-strict order}
\end{definition}

The atomic relation for \emph{plane orientation} used in this work is adopted from the constructive axiom set for Euclidean plane geometry introduced in~\cite{kellison}:
given the three points $a,b,$ and $c$, if the point $a$ lies to the left of the segment $bc$, then the atomic \emph{leftness} relation on points will be denoted by $\Left(a,bc)$. We use the atomic \emph{leftness} relation to define an apartness relation between a point a segment as follows. 
\begin{definition}[Apartness of a point and a segment]The point $a$ lies apart from the segment $bc$ if it is either to the left of the segment $bc$ or to the left of the segment $cb:$ \[ a \# bc := \Left(a,bc) \ \lor \ \Left(a,cb).\] \label{def:lsep}
\end{definition}

\subsection{The Constructive Interpretation of Classical Geometric Relations}\label{sec:classical-relations}
The classical relations of \emph{equivalence},
\emph{collinearity}, \emph{betweenness}, and \emph{congruence} are defined using the atomic relations of \emph{leftness} and \emph{strictly greater than}. In this section, we give the definitions of these relations, and provide the proof of a useful theorem as a simple example of proving stable propositions using constructive logic.

\begin{definition}[Equivalence on points]
The points $a$ and $b$ are \emph{equivalent} if they do not satisfy the binary apartness relation on points (Definition~\ref{def:papart}): \[a \equiv b := \neg a \# b.\] \label{def:eq}
\end{definition}
As is mentioned in Section~\ref{sec:stable}, equivalence and equality do not coincide: while equality on points implies equivalence, equivalence does not imply equality. 

\begin{definition}[Collinearity]
The points $a,b,$ and $c$ are \emph{collinear} if they do not satisfy the apartness relation between a point and a segment: \[\col(abc):= \neg(a \# bc).\] 
\label{def:col}
\end{definition}

\begin{definition}[Betweenness]The point $b$ lies \emph{between} the points $a$ and $c$ if $a,b,$ and $c$ are collinear and the length of the segment $ac$ is not strictly greater than the lengths of $ab$ and $bc$:\[ B(abc) := \col(abc) \ \wedge \ ac \ge ab \ \wedge \ ac \ge bc.\]
\label{def:bet}
\end{definition} 
Note that the above definition coincides with what is referred to as \emph{non-strict} betweenness. That is, the points $a,b,$ and $c$ may be equivalent.

\begin{theorem}[Collinear Cases] Any three collinear points satisfy a weak betweenness relation. 
\begin{align*}
    \forall &a,b,c:\p. \ \col(abc) \Rightarrow \\ &\neg \neg (B(abc) \ \lor \ B(cab)  \lor B(bca) \ \lor a \equiv b \ \lor a \equiv c \ \lor \ b \equiv c).
\end{align*}\label{thm:colcases}
\end{theorem}
\begin{proof}
The stability of the conclusion allows for reasoning by cases on \[\neg \neg (a\#b \ \lor \ \neg a\#b),\] and similarly for $a \#c$ and $b \#c$. Consider the case where $a\#b$, $a\#c$, and $b\#c$. Assume \[\neg (B(abc) \ \lor \ B(cab) \ \lor \ B(bca) \ \lor \ a \equiv b \ \lor a \equiv c \ \lor \ b \equiv c),\] and prove false. Observe that  $\neg B(abc) \ \wedge \ \neg B(cab) \ \wedge \ \neg B(bca)$ follows from the assumption. From $\neg B(abc)$ it follows that \[\neg \neg (a\#bc \ \lor \ ab > ac \ \lor \ bc > ac) .\] Stability of the conclusion allows for elimination of the double negation for each betweenness relation, and expanding the disjunctions results in absurdity. 
\end{proof}

\begin{definition}[Strict Betweenness]The point $b$ lies \emph{strictly between} the points $a$ and $c$ if the point $b$ lies between the points $a$ and $c$, and the points $a,b,$ and $c$ satisfy apartness relations: \[ SB(abc) := B(abc) \ \wedge \ a\# b \ \wedge \ b \# c.\] 
\end{definition}

\begin{definition}[Congruence]The segments $ab$ and $cd$ are \emph{congruent} if they do not satisfy the apartness relation on segment lengths (Definition~\ref{def:lenapart}):\[ab \cong cd := \neg ab \# cd.\]
\label{def:cong}
\end{definition} 

\begin{definition}[Out] The point $p$ lies \emph{out} along the segment $ab$ if it is separated from both $a$ and $b$ and satisfies some weak betweenness relation with $a$ and $b$. Observe that this definition can be viewed as using an constructive interpretation of the classical disjunction used in Definition 6.1 of~\cite{metageo}.  
\begin{align*} out(p,ab) := p \# a \ \wedge \ p \# b \ \wedge \ \neg( \neg B(pab) \ \wedge \ \neg B(pba) )\end{align*}
\label{def:out}
\end{definition}

The universally quantified axioms introduced in Section~\ref{sec:axioms} imply that \emph{collinearity}, \emph{betweenness}, and \emph{congruence} are equivalence relations.

\subsection{Angle Relations}
Our proof of the Steiner-Lehmus theorem required constructive definitions for angle congruence, the sum of two angles, and angle ordering. The following definition of angle congruence is taken from Tarski~\cite{metageo}, but has been modified to use the appropriate constructive relations.
\begin{definition}[Congruent Angles]The angles $abc$ and $xyz$ are \emph{congruent} if the segments of each angle are distinct and there exist points making the corresponding segments of the two angles congruent: 
\begin{align*}
abc &\cong_a xyz  := \\  &a\#b \ \wedge\ b\#c\ \wedge\ x\#y\ \wedge \ y\#z \ \wedge \ \\ &(\exists a',c',x',z':\p. \ B(baa') \ \wedge \ B(bcc') \ \wedge \ B(yxx') \ \wedge \\ &B(yzz') \ \wedge \ ba' \cong yx' \ \wedge \ bc' \cong yz' \ \wedge \ a'c' \cong x'z').
\end{align*}\label{def:anglecong}
\end{definition}
The set of Axioms U, introduced in Section~\ref{sec:axioms}, imply that angle congruence is an equivalence relation. 

\begin{definition}[Sum of two angles]
\begin{align*}
abc \ &\texttt{+} \  xyz = def := \\ &\exists p,p',d',f':\p .\ abc\cong_a dep \  \wedge \ fep \cong_a xyz  \ \wedge \\ &B(ep'p) \ \wedge \ out(edd') \ \wedge \ out(eff') \ \wedge \ SB(d'p'f')
\end{align*}
\label{def:anglesum}
\end{definition}

\begin{figure}[h!]
\centering
     \includegraphics{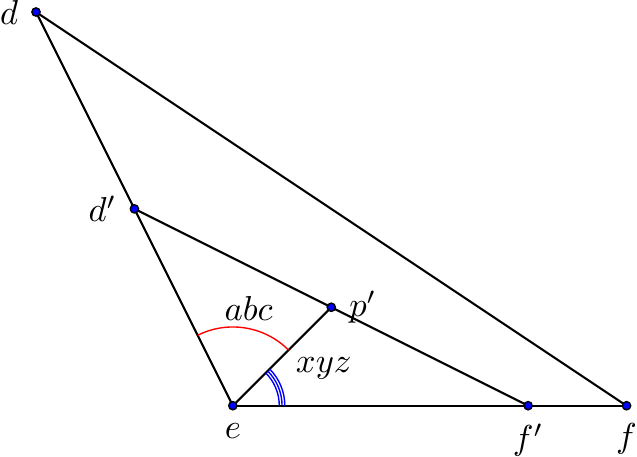}
     \caption{A diagram of Definition~\ref{def:anglesum}: $abc \ \texttt{+} \  xyz = def$ with $p'=p$.}
\end{figure}

\begin{definition}[Angle Inequality]
\begin{align*}
abc  &<_a  xyz  =:
    \neg{}out(y  xz) \ \wedge  \\
     & \exists{} \ p, p', x', z':\p. \
            abc  \cong_a  xyp  \ \wedge \\  &B(yp'p)  \ \wedge \   out(y  xx')  \ \wedge \ out(y  zz')  \ \wedge \\  &\neg B(xyp)  \ \wedge \  B(x'p'z)'  \ \wedge \  p'  \#  z'
\end{align*}
\label{def:angleineq}
\end{definition}

\begin{figure}[htp]
\subfloat[Definition \ref{def:angleineq} for typical angles, with $p = p'$.]{%
  \includegraphics[clip,width=0.6\columnwidth]{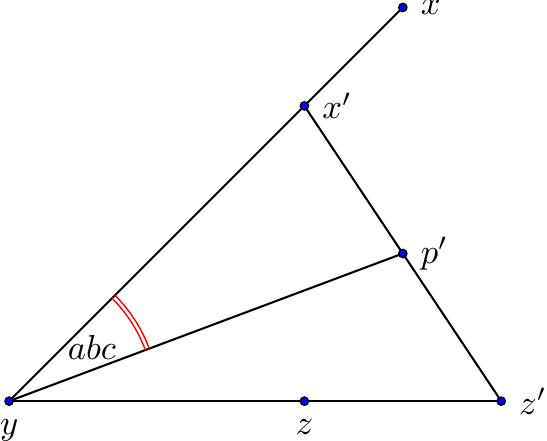}%
}
\hfill
\subfloat[Definition \ref{def:angleineq} for a straight angle.]{%
  \includegraphics[clip,width=0.6\columnwidth]{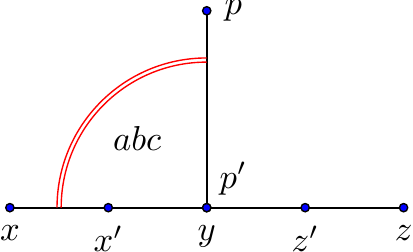}%
}

\caption{Definition \ref{def:angleineq}, $abc  <_a  xyz$.}

\end{figure}

Our axioms imply that angle inequality is a transitive relation for angles satisfying the ternary apartness relation on points (Definition~\ref{def:lsep}).

\subsection{Parallel Segments}
The following definition of parallel segments was essential to our proof of the Steiner-Lehmus theorem. 

\begin{definition}[Parallel Segments]The segments $ab$ and $cd$ are \emph{parallel} if $a\#b$ and $c\#d$ and there do not exist points $x$ and $y$ collinear with $ab$ such that $x$ and $y$ lie on opposite sides of $cd$:
\begin{align*}
ab &\parallel cd  :=  a \# b  \ \wedge \  c \# d \ \wedge \  \neg(\exists x,y : \p. \  \col(xab) \ \wedge \\  &\col(yab) \ \wedge  \Left(x,cd) \ \wedge \  \Left(y, dc)).
\end{align*}
\end{definition}

According to our axioms introduced in the following section, parallelism is a symmetric and reflexive relation but not a transitive relation. Transitivity of parallelism is known to be equivalent to the parallel postulate~\cite{parallelCoq}, which is not an axiom of the theory presented in this paper.

\section{Construction Postulates and Axioms}\label{sec:axioms}
The axioms are introduced here in two separate groups: Axioms U and Axioms C. Axioms U are universally quantified and contain no disjunctions or existential quantifiers. The application of any one of these axioms does not result in a geometric construction. Axioms C are constructor axioms relying on disjunctions and existential quantifiers. As a result, the axioms in group C have a convenient functional reading which may be used in proofs. 

From the Axioms listed below, Axioms \ref{ax:leftcycle}--\ref{ax:colinearlsep} and Axioms \ref{ax:co-trans}--\ref{ax:nontriv} are also used in a previous set of constructive axioms for Euclidean geometry~\cite{kellison}. Unlike the previous system, the relations of betweenness, congruence, and apartness are not primitives of our theory, they are instead defined as described in Section~\ref{sec:classical-relations}. Furthermore, as previously noted, the stability of congruence and betweenness are not taken as axioms in the current work. Finally, the axioms \ref{ax:SC} and \ref{ax:CC} have been simplified to remove the assertion of the existence of redundant points.

\subsection{Universally Quantified Axioms}\label{sec:uax}
\setcounter{axiomu}{0}
\begin{axiomu} 
$\forall{}a,b,c:\p .\    bc  \geq{}  aa$ \label{ax:genull}
\end{axiomu}

\begin{axiomu}$ \forall{}a,b,c,d:\p .\    ab>cd  {}\Rightarrow{}  ab  \geq{}  cd$
\end{axiomu}

\begin{axiomu}$\forall{}a,b,c:\p .\    ba>ac  {}\Rightarrow{}  b  \#  c 
$
\label{ax:gt-to-sep}
\end{axiomu}

\begin{axiomu} $ $\newline  $\forall{}a,b,c,d,e,f:\p .\    ab>cd  {}\Rightarrow{}  cd  \geq{}  ef  {}\Rightarrow{}  ab>ef
$ \label{ax:gt-trans}
\end{axiomu}

\begin{axiomu} $ $\newline
$\forall{}a,b,c,d,e,f:\p .\    ab  \geq{}  cd  {}\Rightarrow{}  cd>ef  {}\Rightarrow{}  ab>ef$
 \label{ax:ge-trans}\end{axiomu}

\begin{axiomu}  
$\forall{}a,b,c:\p .\    B(abc)  {}\Rightarrow{}  b  \#  c  {}\Rightarrow{}  ac>ab$
\label{ax:bet-to-gt}
\end{axiomu}

\begin{axiomu}   $\forall{}a,b,c:\p .\    \Left(a,bc)  {}\Rightarrow{}  \Left(b,ca)
$\label{ax:leftcycle}\end{axiomu}

\begin{axiomu} $\forall{}a,b,c:\p .\    \Left(a,bc)  {}\Rightarrow{}  b  \#  c
$\label{ax:lefttosep}
\end{axiomu}

\begin{axiomu} $\forall{}a,b,c,d:\p .\    B(abd)  {}\Rightarrow{}  B(bcd)  {}\Rightarrow{}  B(abc)$ \label{ax:last}
\end{axiomu}

We take an constructive versions of Tarski's Five-Segment axiom and Upper Dimension axiom~\cite{metageo}.

\begin{axiomu}[Five-Segment] 
\begin{equation*}
\begin{aligned}
\forall a,b,c,d,w,x,y,z:\p .\
              (a  \#  b  \ \wedge \ B(abc)  \ \wedge \ B(wxy)   \ \wedge \ \\ \quad ab  \cong{}  wx  \ \wedge \  bc  \cong{}  xy  \ \wedge \  ad  \cong{}  wz  \ \wedge \  bd  \cong{}  xz)  {}\Rightarrow{} \\  cd  \cong{}  yz \end{aligned}
\end{equation*}\label{ax:5seg}
\end{axiomu}

\begin{axiomu}[Upper Dimension]    
\begin{equation*}
\begin{aligned}\forall{}a,b,c,x,y:\p .\    ax  \cong{}  ay  {}\Rightarrow{}  bx  \cong{}  by  {}\Rightarrow{} \\ cx  \cong{}  cy  {}\Rightarrow{}  x  \#  y  {}\Rightarrow{} \col(abc)
\end{aligned}
\end{equation*}
\end{axiomu}

\begin{axiomu}[Convexity of Leftness] 
\begin{equation*}
\begin{aligned}
\forall{}a,b,x,y,z:\p .\    \Left(x,ab)  \ \wedge     \Left(y,  ab)   \ \wedge  B(xzy)  {}\Rightarrow{} \\  \Left(z,ab)
\end{aligned}
\end{equation*}
\label{ax:leftconvex}
\end{axiomu}

\begin{axiomu}
\[
\forall{}a,b,c,y:\p .\    a  \#  bc  {}\Rightarrow{}  y  \#  b  {}\Rightarrow{}  \col(y ab)  {}\Rightarrow{}  y  \#  bc 
\]
\label{ax:colinearlsep}

\end{axiomu}

\subsection{Construction Postulates}
\setcounter{axiomc}{0}

\begin{axiomc}[Cotransitivity of separated points:] 
\[\forall a,b,c: \p. \ a\#b \Rightarrow  a\#c \  \lor \ b\#c\] \label{ax:co-trans}
\end{axiomc}

\begin{axiomc}[Plane Separation] If the points $u$ and $v$ lie on opposite sides of the segment $ab$, then the point $x$, collinear with $ab$, exists between $u$ and $v$.
\begin{equation*}
\begin{aligned}
\forall a,b,u,v: \p. \ (\ \Left(u,ab) \wedge \ \Left(v,ba) \  \Rightarrow \\ \exists x:\p. \ Col(abx) \ \wedge \ B(uxv))
\end{aligned}
\end{equation*}
\label{ax:planesep}
\end{axiomc}    

\begin{figure}[h!]
\centering
     \includegraphics[scale=1.0]{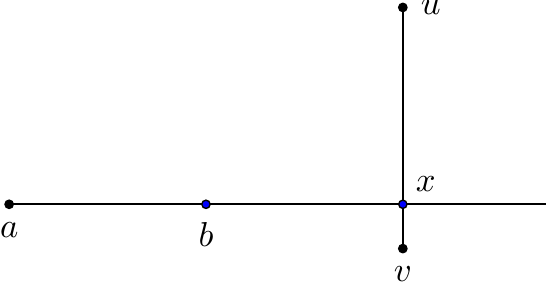}
     \caption{Axiom \ref{ax:planesep}: plane separation.}
     \label{fig:SS}
\end{figure}

\begin{axiomc}[Non-triviality]
\[
\exists a,b: \p.\ a  \#  b
\label{ax:nontriv}
\]
\end{axiomc}    

\begin{axiomc}[Straightedge-Compass] The straight-edge compass axiom constructs a single point of intersection between a circle and a segment (see Figure ~\ref{fig:SC}):
\begin{equation*}
\begin{aligned}
    \forall{}a,b,c,d: \p. \ (a  \#  b  \wedge{}  B(cbd)) \Rightarrow\\
                  \exists{}u:\p .\  cu  \cong{}  cd \  \wedge{} \ B(abu)  \wedge{}  (b  \#  d  {}\Rightarrow{}  b  \#  u)
\end{aligned}
\end{equation*}\label{ax:SC}
\end{axiomc}    

\begin{figure}[h!]
\centering
     \includegraphics[scale=1.0]{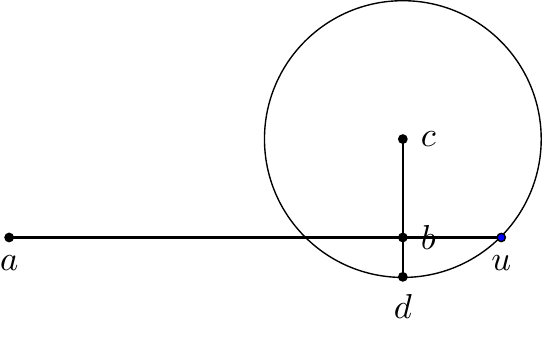}
     \caption{Axiom \ref{ax:SC}, the straight-edge compass construction.}
     \label{fig:SC}
\end{figure}

\begin{axiomc}[Compass-Compass] The compass-compass axioms constructs a single point of intersection between two circles (see Figure ~\ref{fig:CC}): 
\begin{equation*}
\begin{aligned}
\forall{}a,b,c,d:\p. \ a  \#  c \ \wedge \\ (\exists{} \ p,q:\p. \ ab  \cong{}  ap  \wedge{}  cd>cp  \wedge{}  cd  \cong{}  cq  \wedge{}  ab>aq)
 \Rightarrow \\
                  \exists{}u:\p.\  ab  \cong{}  au  \wedge{}  cd  \cong{}  cu  \wedge{}  \Left(u, ac)\label{ax:CC}
\end{aligned}
\end{equation*}
\end{axiomc}    

\begin{figure}[h!]
\centering
     \includegraphics[scale=1.0]{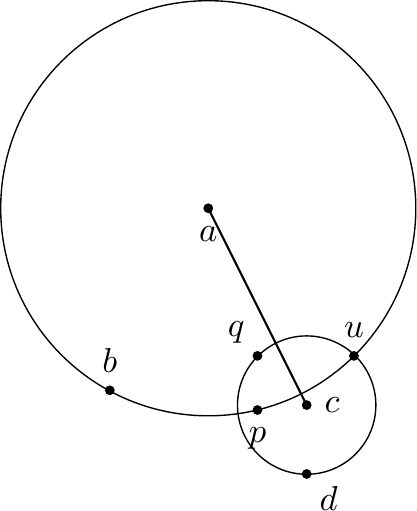}
     \caption{Axiom \ref{ax:CC}, the compass-compass construction.}
     \label{fig:CC}
\end{figure}

\section{The Steiner-Lehmus Theorem}\label{sec:SLT}
The conclusion of the Steiner-Lehmus theorem is stable and so it suffices to prove the double negation of auxiliary theorems with constructive content. Thus, rather than proving a lemma stating that
\begin{center} \emph{from two points along the sides of any triangle, a parallelogram can be constructed such that one side of the parallelogram lies along one side of the triangle,}
\end{center}
we prove the following lemma: 
\begin{lemma}
\begin{equation*}
\begin{aligned}
\forall a,b,c,x,y:\p. \ ( a \#bc \ \wedge \ SB(axb) \ \wedge \ SB(cyb) \Rightarrow \\ \neg \neg (\exists t:\p. \  yt \ || \ ax \ \wedge \ xt \ || \ ay \ \wedge \\ 
ax \cong yt \ \wedge \ xt \cong ay \ \wedge \ t \# bc).
\end{aligned}
\end{equation*}
\label{lem:parallelogram}
\end{lemma}

\begin{proof}
Construct the midpoint $m$ along the segment $xy$ using Euclid I.10 (Theorem~\ref{thm:E10}), and extend the segment $am$ to construct the point $t$ such that $am \cong mt$ by Lemma~\ref{lem:SE}. Now, the angle congruence  $xma \cong_a ymt$ follows from Euclid I.15 (Theorem~\ref{thm:E15}), and the congruence relations $ax \cong yt$ and $xt \cong ay$ follow from Euclid I.4 (Theorem~\ref{thm:E4}) and Axiom~\ref{ax:5seg}, respectively. The angle congruence $axy \cong_a tyx$ then follows by definition, and from Euclid I.27 (Theorem~\ref{thm:E27}) it follows that $ax \parallel yt$ and $xt \parallel ay$. Finally, stability of the conclusion allows for reasoning by cases on $t\#bc$ or $\col(tbc)$.
\\ 

If $\col(tbc)$ then by Lemma~\ref{lem:consuni} the point $t$ must be the point $p$ such that $SB(bpc)$ and $SB(amp)$; $p$ is guaranteed to exist by construction using Lemma~\ref{lem:SOP}. Without loss of generality, from $a\#bc$ assume
$\Left(a, cb)$. From Lemma~\ref{lem:outleft}, it follows that $\Left(c, xa)$. Now, construct the point $q$ by Lemma~\ref{lem:SE} such that $SB(cbq)$ and $SB(ypq)$. It follows from Lemma~\ref{lem:lefttoright} that $\Left(q, ax)$, contradicting  $ax \parallel yt$. 
\end{proof}

\begin{figure}[h!]
\centering
     \includegraphics{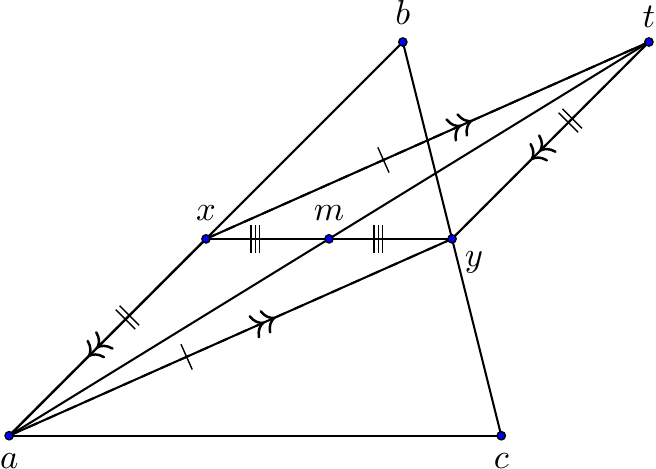}
     \caption{Lemma \ref{lem:parallelogram}}
\end{figure}

\begin{theorem}[Steiner-Lehmus] 

\begin{align*}
\forall a,b,c,x,y:\p. \ ( a \#bc \ \wedge \ SB(axb) \ \wedge \ SB(cyb) \ \wedge \\ ay \cong cx \ \wedge \ xay \cong_a cay \ \wedge \ ycx \cong_a acx \Rightarrow ab \cong cb ). 
\end{align*}
\label{thm:SL}
\end{theorem}
\begin{proof}
Construct the parallelogram $ayxt$ by Lemma \ref{lem:parallelogram}. From Euclid I.5 (Theorem~\ref{thm:E5}) it follows that $xct \cong_a xtc$. The angle sum relations $xty \ \texttt{+} \ ytc \cong_a xtc$ and $xcy\ \texttt{+}\ yct \cong_a xct$ follow by definition from construction of the point $q$ using Axiom~\ref{ax:planesep} such that $SB(qyc)$, $B(tyy)$, $SB(xqt)$, and $B(cqq)$. Now, stability of the conclusion allows for reasoning by cases on $cy >ax$  or $\neg(cy >ax)$. 

If ${cy >ax}$, then ${cy >yt}$ by definition of the parallelogram ${ayxt}$. From Euclid I.25 (Theorem~\ref{thm:E25}) it follows that ${acx <_a cay}$, and therefore ${xcy <_a xty}$. It then follows from Euclid I.18 (Theorem~\ref{thm:E18}) that ${ tcy <_a ytc}$, which, along with Lemma~\ref{lem:anglesumineq} and the angle sum relations ${xty \ \texttt{+} \ ytc \cong_a xtc}$ and ${xcy\ \texttt{+}\ yct \cong_a xct}$, yields the contradiction ${xty <_a xcy}$.

If $\neg(cy >ax)$, it follows that ${\neg \neg (ax >cy \lor ax \cong cy)}$: stability of the conclusion allows for elimination of the double negation, so that we can reason by cases on $ax >cy$ or  $ax \cong cy$. A contradiction is reached for $ax >cy$ by the same reasoning used for $cy >ax$.  

Finally, if  $ax \cong cy$, it follows from Definition ~\ref{def:anglecong} that $xac \cong_a yca$. Theorem ~\ref{thm:E6} then yields $ab \cong cb$, as desired.

\end{proof}

\begin{figure}[h!]
\centering
     \includegraphics[scale=1.0]{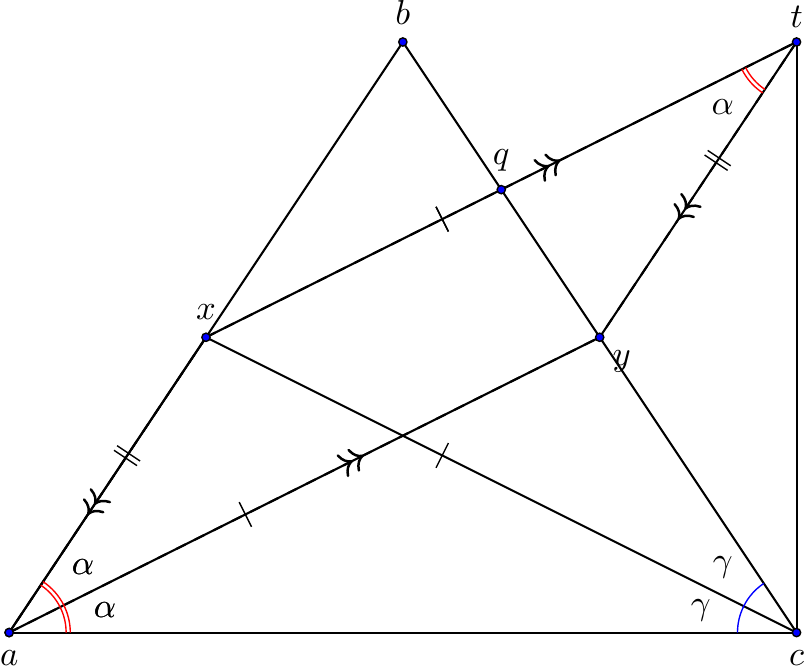}
     \caption{The Steiner-Lehmus Theorem \ref{thm:SL}}
\end{figure}

\section{Essential Auxiliary Theorems}\label{sec:ess_theorems}
This section contains only the statements of the auxiliary theorems used in the proof of the Steiner-Lehmus theorem (Theorem~\ref{thm:SL}) and Lemma~\ref{lem:parallelogram}. The names given to the theorems in this section match their names in the Nuprl library\footnote{\url{http://www.nuprl.org/LibrarySnapshots/Published/Version2/Mathematics/euclidean!plane!geometry/index.html}}. Some definitions used in the Nuprl statement of a theorem may occur unfolded in the following theorem statements for clarity.

\begin{theorem}[geo-extend-exists]
\[\forall q, a, b, c :\p.\ q \# a \Rightarrow \exists x:\p. \ B(qax) \ \wedge \ ax \cong bc.\]
\label{lem:SE}
\end{theorem}

\begin{theorem}[Euclid-Prop4] If two triangles have two sides equal to two sides respectively, and have the angles contained by the equal straight lines equal, then they also have the base equal to the base, the triangle equals the triangle, and the remaining angles equal the remaining angles respectively, namely those opposite the equal sides.
\begin{equation*}
\begin{aligned}
    \forall{}a,b,c,x,y,z:\p. \
    a \#b \ \wedge a \# c \ \wedge b \#c \ \wedge{} \ x\#y \ \wedge{} \ x\#z  \ \wedge{} \\ y\#z 
   \ \wedge{} \  ab  \cong{}  xy\  \wedge{} \ bc  \cong{}  yz\  \wedge{}\  abc  \cong_a  xyz
    {}\Rightarrow{}  \\ ac  \cong{}  xz \ \wedge{} \ bac  \cong_a  yxz \ \wedge{} \ bca  \cong_a  yzx.
\end{aligned}
\end{equation*}
\label{thm:E4}
\end{theorem}

\begin{theorem}[Euclid-Prop5] In isosceles triangles the angles at the base equal one another, and, if the equal straight lines are produced further, then the angles under the base equal one another.

\begin{align*}
    \forall{}a,b,&c,x,y:\p. \
    ab  \cong  ac  \wedge{}  a  \#  bc  \wedge{}  SB(abx) \  \wedge{} \\ &SB(acy)  {}\Rightarrow{}  abc  \cong_a  acb  \wedge{}  xbc  \cong_a  ycb.
\end{align*}
\label{thm:E5}
\end{theorem}

\begin{theorem}[Euclid-Prop6] If in a triangle two angles equal one another, then the sides opposite the equal angles also equal one another.

\begin{align*}
    \forall{}a,b,c:\p. \
  c\#ab \Rightarrow cab \cong_a cba \Rightarrow ca \cong cb.
\end{align*}
\label{thm:E6}
\end{theorem}

\begin{theorem}[Euclid-Prop10] To bisect a given straight line.
\begin{equation*}
\begin{aligned}
    \forall{}a,b:\p. \ a \#b \ \Rightarrow  \\  \exists{}d:\p. \ SB(adb)  \ \wedge{} \  ad  \cong  db.
\end{aligned}
\end{equation*}
\label{thm:E10}
\end{theorem}

\begin{theorem}[vert-angles-congruent]If two straight lines cut one another, then they make the vertical angles equal to one another.
\label{thm:E15}
\begin{equation*}
\begin{aligned}    \forall a,b,c,x,y:\p. \ SB(abx) \ \wedge \ SB(cby) \ \Rightarrow \ abc \cong_a xby.
\end{aligned}
\end{equation*}
\end{theorem}

\begin{theorem}[Euclid-Prop18] In any triangle the angle opposite the greater side is greater.
\begin{align*}
    \forall{}a,b,c:\p. \   a  \#  bc  \ \wedge \  ac  > ab \ {}\Rightarrow{} \ bca  <_a  abc.
\end{align*}
\label{thm:E18}
\end{theorem}

\begin{theorem}[Euclid-Prop25] If two triangles have two sides equal to two sides respectively, but have the base greater than the base, then they also have the one of the angles contained by the equal straight lines greater than the other.

\begin{align*}
    \forall{}a,b,c,d,&e,f:\p. \
    a  \#  bc \ \wedge \  d  \#  ef  \ \wedge \  ab \cong{}  de  \ \wedge \ \\ &ac  \cong{}  df \ \wedge \  bc >ef {}\Rightarrow{}  edf  <_a  bac.
\end{align*}

\label{thm:E25}
\end{theorem}

In the following theorem, the $\Left$ relation is used in the antecedent to capture the notion of ``alternate angles."

\begin{theorem}[Euclid-Prop27] If a straight line falling on two straight lines makes the alternate angles equal to one another, then the straight lines are parallel to one another: 
\begin{equation*}
\begin{aligned} \forall{}a,b,c,d,x,y:\p.\ 
    (\col(xab) \wedge{}  \col(ycd) \ \wedge \  a  \#  b  \ \wedge \\  c  \#  d  \ \wedge \  \Left(a,  yx)  \ \wedge \  \Left(c,  xy)  \ \wedge \  axy  \cong_a {}  cyx)
    {}\Rightarrow{} \\ ab\parallel cd.
\end{aligned}
\end{equation*}\label{thm:E27}
\end{theorem}

\begin{lemma}[geo-intersection-unicity] 
\begin{equation*}
\begin{aligned}     \forall a,b,c,d,p,q : \p. \  \neg \col(abc) \ \wedge \ c\#d \ \wedge \\ \col(abp) \ \wedge \ \col(abq) \ \wedge \ \col(cdp) \ \wedge \ \col(cdq) \Rightarrow p \equiv q .
\end{aligned}
\end{equation*}\label{lem:consuni}
\end{lemma}

\begin{lemma}[left-convex] Given a segment $ab$ and a point $x$ lying to the left of it, the point $y$ lying \emph{out} from $x$ that along the segment $ax$ or $bx$ is in the same half-plane as $x$. 
\begin{equation*}
\begin{aligned}     \forall a,b,x,y: \p. \ \Left(x,ab) \ \wedge \ ( out(axy) \ \lor \ out(bxy))   \Rightarrow \\ \Left(y,ab)
\end{aligned}
\end{equation*}\label{lem:outleft}
\end{lemma}

\begin{lemma}[geo-left-out] Given a segment $ab$ and a point $c$ lying \emph{out} from $b$ along $ab$, if the point $x$ lies to the left of $ab$, then $x$ also lies to the left of $ac$.
\begin{align*}
    \forall a,b,c,x: \p. \ \Left(x,ab) \ \wedge \ out(abc) \Rightarrow  \Left(x,ac)
\end{align*}
\label{lem:leftout2}
\end{lemma}

\begin{lemma}[strict-between-left-right]
\begin{equation*}
\begin{aligned} 
\forall a,b,c,x,y: \p. \  \Left(x,ab) \ \wedge \ \col(abc) \ \wedge \  SB(xcy) \Rightarrow \\ \Left(y,ba)
\end{aligned}
\end{equation*}
\label{lem:lefttoright}
\end{lemma}

\begin{theorem}[outer-pasch-strict]
\begin{align*}
\forall a,b,c,&x,q: \p. \ x\#bq \ \wedge \  SB(bqc) \ \wedge \ SB(qxa) \Rightarrow \\ &\exists p:\p. \ SB(bxp) \ \wedge \ SB(cpa). 
\end{align*}
\label{lem:SOP}
\end{theorem}

\begin{lemma}[hp-angle-sum-lt4] If the sum of the strict angles $abc$ and $xyz$ is equal to the sum of the strict angles $a'b'c'$ and $x'y'z'$, and $x'y'z'$ is less than $xyz$, then it must be the case that the angle $abc$ is less than the angle $a'b'c'$.
\begin{equation*}
\begin{aligned} \forall{}a,b,c,x,y,z,i,j,k:\p.\\ \forall a',b',c',x',y',z',i',j',k':\p.\\
    abc  +  xyz  \cong{}  ijk
    \ \wedge \  a'b'c'  +  x'y'z'  \cong{}  i'j'k' 
    \ \wedge \ \\ ijk  \cong_a  i'j'k'
    \ \wedge \   a'  \#  b'c'
    \ \wedge \  x'  \# y'z'
    \ \wedge \  x  \#  yz
    \ \wedge \  i  \# jk
    \ \wedge \ \\  x'y'z'  <  xyz {}\Rightarrow{}
      abc  <  a'b'c'.
\end{aligned}
\end{equation*}
\label{lem:anglesumineq}
\end{lemma}

\section{A Model on the Constructive Reals}\label{sec:model}
The soundness of our axioms with respect to the Nuprl implementation of the constructive reals~\cite{Bickford1} is implied by the following interpretations of our primitives\footnote{The proofs of soundness for the Axiom sets U and C can be found at \url{http://www.nuprl.org/LibrarySnapshots/Published/Version2/Mathematics/reals!model!euclidean!geometry} as the theorems \\ \emph{r2-basic-geo-axioms} and \emph{r2-eu\_wf}, respectively. }. 

\begin{definition} If $x \in \mathbb{R}$ is the length of the segment $ab$ and $y \in \mathbb{R}$ is the length of the segment $cd$, $x$ is \emph{strictly greater} than $y$ if and only if there exists a natural number $n$ such that the $n$th rational terms of $x$ and $y$ differ by more than four:\begin{align*}x>_{\mathbb{R}}y : = \exists n \in \mathbb{N}. \ x(n) >_{\mathbb{Q}} y(n) + 4. \label{def:consR}\end{align*}
Note that the ordering relation $>_{\mathbb{Q}}$ on the rational numbers is decidable (Definition~\ref{def:decidable}) while the ordering relation $>_{\mathbb{R}}$ on the constructive reals is not. 
\end{definition} 

\begin{definition} Given the real coordinates $(x_0, y_0, 1)$, \newline $(x_1, y_1, 1)$, $(x_2, y_2, 1)$ of the points $a,b$ and $c$, respectively, the point $a$ lies \emph{ left of} the segment $bc$ if and only if the determinant of the matrix formed by the points $a,b$ and $c$ is strictly positive: 
\begin{align*}\Left(a,bc) : = 
\begin{vmatrix}
x_0 &y_0 & 1\\ x_1 & y_1 &1\\x_2 & y_2 & 1
\end{vmatrix}
>_{\mathbb{R}}0.\end{align*}
\end{definition}

The soundness of our axioms with respect to the constructive reals provides additional assurance that our axioms do not use \emph{reductio ad absurdum}. Although the axioms presented in this paper differ (as described in Section ~\ref{sec:axioms}) from those presented in ~\cite{kellison},
only minor modifications were necessary for the soundness proofs.

Finally, while the constructive real model for our axioms guarantees that a direct proof of the Steiner-Lehmus theorem exists in the constructive reals, it says nothing about the existence of direct proof in the classical reals.


\section{Conclusion}
We have introduced here for the first time a proof of the Steiner-Lehmus theorem that is entirely absent of the use of \emph{reductio ad absurdum} and can therefore be considered \emph{fully direct}. This theorem was proved in the constructive logic of the Nuprl proof assistant using a novel axiomatization of Euclidean plane geometry without the parallel postulate. The crux of the proof is the realization that congruence in constructive geometry is a \emph{stable relation}, and that the proof of a stable relation permits double negation elimination and therefore also case distinctions.

Finally, we conclude by addressing the suggestion that the many years of failed attempts to find a direct proof of the Steiner-Lehmus theorem was cause to celebrate the indispensability of \emph{reductio ad absurdum}. In particular, a discussion of the Steiner-Lehmus theorem given in a geometry textbook by Coxeter and Greitzer~\cite{coxeter} includes the popular quote of G. H. Hardy~\cite{Hardy}: \emph{Reductio ad absurdum, which Euclid loved so much, is
one of a mathematician's finest weapons.} We instead propose the following:
\begin{center}
\emph{Double negation is
one of a mathematician's finest weapons, and a proof assistant one of her most steadfast companions.}
\end{center}
 \balance
\section*{Acknowledgements} The author would like to thank Robert Constable and Mark Bickford for being patient teachers of intuitionism and type theory, and would also like to acknowledge their contributions to 
the development of constructive Euclidean plane geometry in Nuprl. The author would also like to thank Andrew Appel for his helpful feedback and encouragement. Finally, the author would like to thank the anonymous reviewers whose comments helped to improve this paper. 

The author acknowledges that this material is based upon work supported by the U.S. Department of Energy, Office
of Science, Office of Advanced Scientific Computing Research, Department of Energy Computational
Science Graduate Fellowship under Award Number DE-SC0021110. This report was prepared as an account of work sponsored by an agency of the United
States Government. Neither the United States Government nor any agency thereof, nor any of their
employees, makes any warranty, express or implied, or assumes any legal liability or responsibility for the
accuracy, completeness, or usefulness of any information, apparatus, product, or process disclosed, or
represents that its use would not infringe privately owned rights. Reference herein to any specific
commercial product, process, or service by trade name, trademark, manufacturer, or otherwise does not
necessarily constitute or imply its endorsement, recommendation, or favoring by the United States
Government or any agency thereof. The views and opinions of authors expressed herein do not
necessarily state or reflect those of the United States Government or any agency thereof.

 \bibliographystyle{ACM-Reference-Format}
 \bibliography{bib.bib}

\end{document}